\documentclass[11pt]{article}

\usepackage{amsmath,amsthm,amssymb,amsfonts,verbatim}
\usepackage{fullpage}

\usepackage[dvipsnames,usenames]{color}

 \providecommand{\F}{\mathbb{F}}

\title{{\bf An Improvement on the Hasse-Weil Bound  and applications to Character Sums, Cryptography and Coding }}
\author{Ronald Cramer\thanks{CWI, Amsterdam \& Mathematical Institute, Leiden University, The Netherlands. {\tt ronald.cramer@cwi.nl}} \and Chaoping Xing\thanks{Division of Mathematical Sciences, School of Physical \&  Mathematical Sciences, Nanyang Technological University, Singapore. {\tt xingcp@ntu.edu.sg}}}
\date{}

\newtheorem{theorem}{Theorem}[section]
\newtheorem{cor}[theorem]{Corollary}
\newtheorem{prop}[theorem]{Proposition}
\newtheorem{lem}[theorem]{Lemma}
\newtheorem{definition}[theorem]{Definition}
\newtheorem{ex}[theorem]{Example}

\theoremstyle{remark}

\renewcommand{\epsilon}{\varepsilon}
\renewcommand{\le}{\leqslant}
\renewcommand{\ge}{\geqslant}

\def\CC{\mathbb{C}}
\def\RR{\mathbb{R}}
\def\ZZ{\mathbb{Z}}
\def\QQ{\mathbb{Q}}

\def\EE{\mathbb{E}}
\def \X {\mathcal{X}}

\def \oF {\overline{\F}}

\def \mJ {\mathcal{J}}

\def \mP {\mathcal{P}}

\def \mX {\mathcal{X}}
\def \mY {\mathcal{Y}}

\def\cL{{\mathcal L}}

\def\Pin{{P_{\infty}}}
\def\Tr{{\rm Tr}}

\def\Trp{{\Tr_{\mP}}}
\def\NL{{\rm NL}}
\def\wt{{\rm wt}}

\newcommand{\Ga}{\alpha}
\newcommand{\Gb}{\beta}
\newcommand{\Gc}{\chi}
\newcommand{\Gg}{\gamma}

\newcommand{\Go}{\omega}

\def\BCH{{\rm BCH}}

\def \tL{{\tilde{L}}}

\def \epf {\hfill $\Box$\\ \par}

\begin{document}
\maketitle

\thispagestyle{empty}

\begin{abstract}
The Hasse-Weil bound is a deep result in mathematics and has found wide applications in mathematics, theoretical computer science, information theory etc. In general, the bound is tight and cannot be improved. However, for some special families of curves the bound could be improved substantially. In this paper, we focus on the Hasse-Weil bound for the curve defined by $y^p-y=f(x)$ over the finite field $\F_q$, where $p$ is the characteristic of $\F_q$. Recently,  Kaufman and Lovett \cite[FOCS2011]{KL11} showed that  the Hasse-Weil bound   can be improved for this family of curves with $f(x)=g(x)+h(x)$, where $g(x)$ is a polynomial of degree $\ll \sqrt{q}$ and $h(x)$ is a sparse polynomial of arbitrary degree but bounded weight degree. The other recent improvement by Rojas-Leon  and Wan \cite[Math. Ann. 2011]{RW11} shows that an extra $\sqrt{p}$ can be removed for this family of curves  if $p$ is very large compared with polynomial  degree  of $f(x)$ and $\log_pq$.

In this paper, we show that the Hasse-Weil bound for this special family of curves can be improved if $q=p^n$ with odd $n$ which is the same case where Serre \cite{Se85} improved the Hasse-Weil bound. However, our improvement is greater than Serre's one for this special family of curves.  Furthermore, our improvement works for small $p$ as well compared with the requirement of large $p$ by Rojas-Leon  and Wan. In addition, our improvement finds interesting applications to character sums, cryptography and coding theory.  The key idea behind is that this curve has the Hasse-Witt invariant $0$ and we show that the Hasse-Weil bound can be improved for any curves with the Hasse-Witt invariant $0$. The main tool used in our proof involves Newton polygon and some results in algebraic geometry.
\end{abstract}

\newpage
\section{Introduction}
Let $\chi$ be a nontrivial additive character from $\F_q$ to the nonzero complex $\CC^*$. The Weil bound for character sums states that, when the degree $m$ of a polynomial $f(x)$ is co-prime with $q$, then $|\EE(\chi(f(x)))|\le (m-1)/\sqrt{q}$.
The Weil bound for character sums can be derived from  the Hasse-Weil bound of a special family of algebraic curves, i.e., $y^p-y=f(x)$, where $p$ is the characteristic of $\F_q$. Therefore, any improvement on the Hasse-Weil bound of this family could lead to an improvement on the Weil bound for character sums.
 The Hasse-Weil bound is a deep result in mathematics and has found wide applications in mathematics, theoretical computer science, information theory etc. In general, the bound is tight and cannot be improved. However, in some special cases the bound could be improved substantially. For instance, when the ground filed size $q$ is small, the bound could be improved up to half of it. There have been various improvements on the Hasse-Weil bound. One of the most famous improvements is the Weil-Serre bound \cite{Se85}. Recently,  Kaufman and Lovett \cite[FOCS2011]{KL11} showed that  the Weil bound for character sums  can be improved for $f(x)=g(x)+h(x)$, where $g(x)$ is a polynomial of degree $\ll \sqrt{q}$ and $h(x)$ is a sparse polynomial of arbitrary degree but bounded weight degree. The other recent improvement by Rojas-Leon  and Wan \cite[Math. Ann. 2011]{RW11} shows that an extra $\sqrt{p}$ can be removed if $p$ is very large compared with polynomial  degree $m$ of $f(x)$ and $\log_pq$.

\subsection{Our result}

From now on in this paper, we assume that $q=p^{n}$ for a prime $p$ and an odd integer $n\ge 3$.
Consider the trace map $\Tr$ from $\F_q$ to $\F_p$ defined by $\Ga\mapsto\sum_{i=0}^{n-1}\Ga^{p^i}$. It is easy to see that $\Tr(\Ga^{pt})=\Tr(\Ga^t)$. This implies that, for a polynomial $f(x)\in\F_q[x]$, one can find a polynomial $g(x)=\sum_{i=0}^mg_ix^i\in\F_q[x]$ such that $g_i$ are zero whenever $i\equiv 0\bmod{p}$ and $\Tr(f(\Ga))=\Tr(g(\Ga))$. Thus, we only consider those polynomials with nonzero term $x^i$, where $\gcd(i,p)=1$.

Now let $f(x)$ be a polynomial of degree $m\ge 1$ over $\F_q$.
Without loss of generality, we may assume that $\gcd (m,p)=1$. We consider the cardinality
%Fix an element $\Gk\in\F_q$ such that $y^p+\Gk y$ has $p$ distinct roots in $\F_q$.
of the set
\begin{equation}\label{eq2.0}Z_f:=\{\Ga\in\F_q:\; \Tr(f(\Ga))=0\}.  \end{equation}
Then the Weil bound shows that
\begin{equation}\label{eq:e0}\left||Z_f|-\frac{q}{p}\right|\le \frac{(p-1)(m-1)\sqrt{q}}p.  \end{equation}

The main result of this paper is given below.
\begin{theorem}[MAIN RESULT]\label{1.1}
Let $g=(m-1)(p-1)/2$, then one has
\begin{equation}\label{eq2.5}\left||Z_f|-\frac qp\right|\le\left\{\begin{array}{ll} p^{\lceil n/g\rceil-1}\left\lfloor\frac{g\lfloor2\sqrt{q}\rfloor}{p^{\lceil n/g\rceil}}\right\rfloor& \mbox{if $g>1$}\\
p^{(n-1)/2}\left\lfloor\frac{\lfloor2\sqrt{q}\rfloor}{p^{ (n+1)/2}}\right\rfloor & \mbox{if $g=1$}.\end{array}\right.\end{equation}
\end{theorem}

One could not see how largely the bound (\ref{eq:e0}) is improved in Theorem \ref{1.1}. In fact, the improvement could be substantial. We refer this to Example \ref{ex2.2} and Sections 2.3 and 2.4 for numerical illustration.

\subsection{Our technique}
Our  technique for the improvement is through an improvement of the Hasse-Weil bound for the algebraic curve. More precisely speaking, we show the improvement through three steps: (i) show that the Weil bound for character sums is derived from the Hasse-Weil bound of an algebraic curve; (ii) prove that the Hasse-Weil bound for the algebraic curve with the Hasse-Witt invariant $0$ can be improved; (iii) show by the Deuring-Shafarevich Theorem that the curve $y^p-y=f(x)$ has the Hasse-Witt invariant $0$. Consequently, we obtain an improvement on the Weil bound for character sums.

Among the above three steps, the critical one is to show  an improvement on the Hasse-Weil bound for the algebraic curve with the Hasse-Witt invariant $0$. In order to achieve this, we analyze the Newton polygon of the characteristic polynomial of an abelian verity with Hasse-Witt invariant $0$. Then we employ some fundamental results on factorization of the characteristic polynomial to obtain the desired result.
\subsection{Comparison}
We mainly compare our improvement with those by Serre \cite{Se85},  Kaufman-Lovett \cite[FOCS2011]{KL11} and Rojas-Leon-Wan \cite[Math. Ann. 2011]{RW11}.

The improvement by Serre applies to arbitrary algebraic curve over $\F_{p^n}$ with odd $n$, while our improvement only applies to the curve $y^p-y=f(x)$ over $\F_{p^n}$ with odd $n$. On the other hand, our improvement is  even greater than the one by Serre for  this special family of curves.. The method by Serre mainly employs some properties of algebraic numbers.

The improvement by Kaufman and Lovett works for those polynomial with degree bigger than $\sqrt{q}$. Thus, the scenario is different. The main idea by Kaufman-Lovett uses Deligne Theorem on multivariate polynomials.

The improvement by Rojas-Leon and Wan shows that   an extra $\sqrt{p}$ can be removed if $p$ is very large compared with polynomial  degree $m$ of $f(x)$ and $\log_pq$. However, our improvement works for any characteristic $p$ including $p=2$ (see numerical examples of Section 2). The idea of Rojas-Leon and Wan involves moment $L$-functions and Katz's  work on $\ell$-adic monodromy calculations.

\subsection{Organization}
Section 2 presents the main result and some applications to character sum, cryptography and coding theory. We also show that the  Weil bound for character sums can be derived from the Hasse-Weil bound of $y^p-y=f(x)$  in Section 2. We outline the proof of the main result in Section 3 without requirement on knowledge of algebraic geometry.  In Appendix, we show the detailed proof of the main result with requirement on knowledge of algebraic geometry.

\section{Main result and applications}
\subsection{Main result}
The Hasse-Weil bound \cite{Weil48} provides an upper bound on the number of points of an algebraic curve over a finite field $\F_q$ in terms of its genus and the ground field size $q$. The bound was improved by Serre \cite{Se85} when $q$ is not a square. The refined bound by Serre is now called the Weil-Serre bound.  In this section, we show that the Weil-Serre bound can be further improved for a class of curves arising from trace when $q$ is not a square. Furthermore, several applications are provided for this improvement.

It is a well-known fact that the cardinality of $Z_f$ in (\ref{eq2.0}) is equal to $(N_f-1)/p$, where $N_f$ stands for the number of  the $\F_q$-rational points on the Artin-Schreier type curve defined by
\begin{equation}\label{eq2.1} \mX_{f}:\quad y^p- y=f(x).\end{equation}
Note that the set of the $\F_q$-rational points on $\mX_f$ is $\{\Pin\}\cup\{(\Ga,\Gb)\in\F_q^2:\; \Gb^p-\Gb=f(\Ga)\}$, where $\Pin$ stands for the ``points at infinity".

To see this, we note that  $\Pin$ can be discarded towards counting the cardinality of $Z_f$. Furthermore, an element $\Ga$ belongs to $ Z_f$ if and only if there are $p$ elements $\Gb\in\F_q$ such that $(\Ga,\Gb)$ are solutions of the equation (\ref{eq2.1}).

When applying the Weil bound \cite{Weil48}  to the curve $\mX_f$,   we have
\begin{equation}\label{eq2.2}
|N_f-q-1|\le 2g\sqrt{q},
\end{equation}
where $g=(m-1)(p-1)/2$ is the {\it genus} of $\mX_f$. Serre \cite{Se85} improved the above bound to the following Weil-Serre bound
\begin{equation}\label{eq2.3}
|N_f-q-1|\le g\lfloor2\sqrt{q}\rfloor.
\end{equation}
 It seems that the bound (\ref{eq2.3}) is just a small improvement on the Weil bound. However, the improvement could be substantial if $m$ is big. For instance, $q=32$ and $m=2001$, then the Weil bound gives $|N_f-q-1|\le 11,313$, while the Weil-Serre bound gives $|N_f-q-1|\le 11,000$.  %\par (ii) Note that the Weil-Serre bound applies to an arbitary algebraic curve. The bound (\ref{eq2.3}) is an appliaction of the Weil-Serre bound to the curve $\mX_f$.

Applying the Weil-Serre bound, we get the following bound on the cardinality $|Z_f|$
\begin{equation}\label{eq2.4}
\left||Z_f|-\frac qp\right|\le \left\lfloor \frac{g\lfloor2\sqrt{q}\rfloor}p\right\rfloor.
\end{equation}
In this paper, we show that the bound (\ref{eq2.4}) can be further improved. 
We repeat the main result of this paper is that improves the above bound (\ref{eq2.4}) as follows.

\begin{theorem}[MAIN RESULT]\label{2.1}
Let $g=(m-1)(p-1)/2$, then one has
\begin{equation}\label{eq2.5}\left||Z_f|-\frac qp\right|\le\left\{\begin{array}{ll} p^{\lceil n/g\rceil-1}\left\lfloor\frac{g\lfloor2\sqrt{q}\rfloor}{p^{\lceil n/g\rceil}}\right\rfloor& \mbox{if $g>1$}\\
p^{(n-1)/2}\left\lfloor\frac{\lfloor2\sqrt{q}\rfloor}{p^{ (n+1)/2}}\right\rfloor & \mbox{if $g=1$}.\end{array}\right.\end{equation}
\end{theorem}
The proof of Theorem \ref{2.1}  involves   algebraic geometry and algebraic number theory. We will outline the proof in the next section in a more elementary way and provide some details in Appendix.

From the formula of (\ref{eq2.5}),  one cannot see the big difference between the Weil-Serre bound (\ref{eq2.4}) and our bound (\ref{eq2.5}). Let us use some examples to illustrate improvement.

\begin{ex}\label{ex2.2}{\rm \begin{itemize}
\item[(i)] Let $q=2^7=128$ and $m=\deg(f)=5$, then the Weil-Serre bound (\ref{eq2.4}) gives $\left||Z_f|-64\right|\le 22$, while  the bound  (\ref{eq2.5}) gives $\left||Z_f|-64\right|\le 16$.

\item[(ii)] Let $q=3^5=243$ and $m=\deg(f)=5$, then the Weil-Serre bound (\ref{eq2.4}) gives $\left||Z_f|-81\right|\le 20$, while  the bound  (\ref{eq2.5}) gives $\left||Z_f|-81\right|\le 18$.

    \item[(iii)] Finally, we look at an example of relatively large parameters. Let $q=2^{21}$ and $m=\deg(f)=5$, then the Weil-Serre bound (\ref{eq2.4}) gives $\left||Z_f|-2^{20}\right|\le 2896$, while the bound  (\ref{eq2.5}) gives $\left||Z_f|-2^{20}\right|\le 2048$.
\end{itemize}
}\end{ex}
From Example \ref{ex2.2}(iii), we can see that, for some parameters, there could be a big difference between the Weil-Serre bound (\ref{eq2.4}) and our bound (\ref{eq2.5}).

\subsection{Improvement on the Weil bound for character sum}
It is well-known that every nontrivial additive character from $\F_q$ to $\CC^*$ can be represented  by $\Gc_{\Gb}(x):=\exp(2\pi i\Tr(\Gb x)/p)$ for some $\beta\in\F_q^*$ (see \cite{LN83}). Let $f(x)$ be a polynomial of degree $m$ over $\F_q$ with $\gcd(m,p)=1$ and put $g=(m-1)(p-1)/2$. Then  $\Gc_{\Gb}(f(x))$  is uniformly distributed  when $\sqrt{q}\gg m$. More precisely speaking, the Weil bound gives
\begin{equation}\label{eq:b01} |\EE(\chi(f(x)))|\le \frac{m-1}{\sqrt{q}}.\end{equation}
Furthermore, applying the Weil-Serre bound (\ref{eq2.4}) gives
\begin{equation}\label{eq2.7}|\EE_{x\in\F_q}(\Gc_{\Gb}(f(x)))|\le \frac1q\left\lfloor\frac{g\lfloor 2\sqrt{q}\rfloor}p\right\rfloor,\end{equation}
where $\EE_{x\in\F_q}(\Gc_{\Gb}(f(x)))$ is the expectation of $\Gc_{\Gb}(f(x))$ defined by
\begin{equation}\label{eq2.8}
\EE_{x\in\F_q}(\Gc_{\Gb}(f(x)))=\frac 1q\sum_{\Ga\in\F_q} \Gc_{\Gb}(f(\Ga)) =\sum_{a\in\F_p}\frac{M_{\Gb f}(a)}{q}\exp(2\pi ai/p)
\end{equation}
with $M_{\Gb f}(a)$ being the cardinality of the set $\{\Ga\in\F_q:\; \Tr({\Gb}f(\Ga))=a\}$. Note that $M_{\Gb f}(0)=|Z_{\Gb f}|$.
Now we apply the bound  (\ref{eq2.5}) to get an improvement to the bound (\ref{eq2.7}).

\begin{theorem}\label{2.3}
Let $g=(m-1)(p-1)/2$, then one has
\begin{equation}\label{eq2.9}|\EE_{\Ga\in\F_q}(\Gc_{\Gb}(f(\Ga)))|\le\left\{\begin{array}{ll} p^{\lceil n/g\rceil-n}\left\lfloor\frac{g\lfloor2\sqrt{q}\rfloor}{p^{\lceil n/g\rceil}}\right\rfloor& \mbox{if $g>1$,}\\
p^{(n+1)/2-n}\left\lfloor\frac{\lfloor2\sqrt{q}\rfloor}{p^{ (n+1)/2}}\right\rfloor & \mbox{if $g=1$}\end{array}\right.\end{equation}
for any nonzero element $\Gb\in\F_q$.
\end{theorem}
\begin{proof} For each $a\in \F_p$, choose $\Gg_a\in\F_q$ such that $\Tr(\Gg_a)=a$. Then we have $\Tr(\Gb f(\Ga))=a$ if and only if $\Tr(\Gb f(\Ga)-\Gg_a)=0$. This implies that $M_{\Gb f}(a)=M_{\Gb f-\Gg_a}(0)=|Z_{\Gb f-\Gg_a}|$. Thus, $M_{\Gb f}(a)$ also satisfies the bound (\ref{eq2.5}). Hence,
\begin{eqnarray*}
|\EE_{\Ga\in\F_q}(\Gc_{\Gb}(f(\Ga)))|&=&\left|\sum_{a\in\F_p}\frac{M_{\Gb f}(a)}{q}\exp(2\pi ai/p)\right|
=\frac 1q\left|\sum_{a\in\F_p}\left(M_{\Gb f}(a)-\frac qp\right)\exp(2\pi ai/p)\right|\\
&\le &\frac1q\sum_{a\in\F_p}\left|\left(M_{\Gb f}(a)-\frac qp\right)\exp(2\pi ai/p)\right|
=\frac1q\sum_{a\in\F_p}\left||Z_{\Gb f-\Gg_a}|-\frac qp\right|.
\end{eqnarray*}
The desired result follows from Theorem \ref{2.1}.
\end{proof}
We illustrate our improvement by consider some small examples.
\begin{ex}{\rm Let us consider character sums for polynomials of degree $3$ and $5$, respectively.
\begin{itemize}
\item[(i)] Let $q=2^n$ with odd $n$ and let $f(x)$ be a polynomial of degree $3$. By Theorem \ref{2.3}, for any nontrivial additive character $\Gc$ from $\F_{2^n}$ to $\CC^*$, one has
$\left|\sum_{\Ga\in\F_q} \Gc(f(\Ga))\right|\le 2^{(n+1)/2}.$
The Weil bound (\ref{eq:b01}) gives an upper bound $2^{(n+2)/2}$, while the Weil-Serre bound (\ref{eq2.7}) gives an upper bound between $2^{(n+1)/2}$ and $2^{(n+2)/2}$.
\item[(ii)] Let $q=2^n$ with odd $n$ and let $f(x)$ be a polynomial of degree $5$. By Theorem \ref{2.3}, for any nontrivial additive character $\Gc$ from $\F_{2^n}$ to $\CC^*$, one has
$\left|\sum_{\Ga\in\F_q} \Gc(f(\Ga))\right|\le 2^{(n+3)/2}.$
The Weil bound (\ref{eq:b01}) gives an upper bound $2^{(n+4)/2}$, while the Weil-Serre bound (\ref{eq2.7}) gives an upper bound between $2^{(n+3)/2}$ and $2^{(n+4)/2}$.
\end{itemize}
}\end{ex}
\subsection{Application to cryptography}
In streamcipher, nonlinearity of a function $f(x)$ from $\F_{2^n}$ to  $\F_{2^n}$ is an important measure \cite{C10}. The nonlinearity of a function $f(x)$ is defined as follows.

The Walsh transfer $W_f$ of $f(x)$ is defined by
\begin{equation}\label{eq6.1} W_f:\quad \F_{2^n}\times \F_{2^n}\rightarrow \CC;\qquad (a,b)\mapsto\sum_{x\in\F_{2^n}}(-1)^{af(x)+bx}.
\end{equation}
Then the Walsh spectrum of $f$ is the image set $\{W_f(a,b):\; a\in\F_{2^n}^*, \ b\in \F_{2^n}\}$.
The nonlinearity, denoted by $\NL(f)$, of $f(x)$ is defined by
\begin{equation}\label{eq6.2} \NL(f):=2^{n-1}-\frac12\max_{(a,b)\in \F_{2^n}^*\times \F_{2^n}}|W_f(a,b)|.
\end{equation}
In fact, the Walsh transformation $W_f(a,b)$ is nothing but  $2$ times the expectation of $\Gc_{1}(af(x)+bx)$, i.e., $W_f(a,b)=2\EE_{x\in\F_{2^n}}(\Gc_{1}(af(x)+bx))$. Thus, The nonlinearity  $\NL(f)$ of $f(x)$ can be expressed in terms of the  expectation of $\Gc_{1}(af(x)+bx)$, i.e.,
\begin{equation}\label{eq6.3} \NL(f)=2^{n-1}-\max_{(a,b)\in \F_{2^n}^*\times \F_{2^n}}|\EE_{x\in\F_{2^n}}(\Gc_{1}(af(x)+bx))|.
\end{equation}
For an odd $n$, it can be proved that $\NL(f)$ is upper-bounded by $2^{n-1}-2^{(n-1)/2}$ \cite{C10}. Thus, if $\deg(f)=3$, i.e., the curves defined in (\ref{eq2.1}) is an elliptic curve, then by Theorem \ref{2.1} one has
\begin{eqnarray*}
\NL(f)&=&2^{n-1}-\max_{(a,b)\in \F_{2^n}^*\times \F_{2^n}}|\EE_{x\in\F_{2^n}}(\Gc_{1}(af(x)+bx))|\\
&\ge &2^{n-1}-2^{(n-1)/2}\left\lfloor\frac{\lfloor2^{(n+2)/2}\rfloor}{2^{ (n+1)/2}}\right\rfloor=2^{n-1}-2^{(n-1)/2}.
\end{eqnarray*}
Hence,  by the upper bound, we have $\NL(f)=2^{n-1}-2^{(n-1)/2}$ for polynomials $f$ of degree $3$, i.e., polynomials of degree $3$ produce the highest   nonlinearity.
Previously,   polynomials of degree $3$ are usual candidates for functions with highest nonlinearity. Now let us look a polynomial $f$ of degree $5$ and by the bound (\ref{eq2.9})
\begin{eqnarray*}
\NL(f)&=&2^{n-1}-\max_{(a,b)\in \F_{2^n}^*\times \F_{2^n}}|\EE_{x\in\F_{2^n}}(\Gc_{1}(af(x)+bx))|\\
&\ge &2^{n-1}-2^{(n-1)/2-1}\left\lfloor\frac{2\lfloor2^{(n+2)/2}\rfloor}{2^{ (n+1)/2}}\right\rfloor=2^{n-1}-2^{(n-1)/2}.
\end{eqnarray*}
Thus,  polynomials of degree $5$  also produce functions with the highest nonlinearity. This enlarges the pool of functions with  the highest nonlinearity when people search for functions with the highest  nonlinearity together with other cryptographic properties such as algebraic degree, algebraic immunity, etc. \cite{C10}.

\subsection{Application to coding}
Many codes such as BCH codes, classical Goppa codes and Reed-Muller codes can be realized as trace codes. In fact, every cyclic code can be represented as a trace code in a natural way \cite{Wo88}.

Let $\mP:=\{\Ga_1,\Ga_2,\dots,\Ga_N\}$ be a subset of $\F_q$ of cardinality $N$. For a polynomial $f(x)\in\F_q[x]$, we denote by $\Trp(f)$ the vector
$(\Tr(f(\Ga_1)),\Tr(f(\Ga_2)),\dots,\Tr(f(\Ga_N)))$. For an $\F_p$-subspace $V$ of $\F_q[x]$, we denote by $\Trp(V)$ the trace code $\{\Trp(f):\; f\in V\}$. Let us warm up by looking at a small example first.
\begin{ex}\label{2.4}{\rm
Consider the binary code $\Trp(V)$, where $\mP$ consists of all elements in $\F_{128}$, and $V=\{f(x)\in\F_{128}[x]:\; \deg(f)\le 5\}$. note that $\Trp(x)=\Trp(x^2)=\Trp(x^4)$. Then it is easy to see that $\Trp(V)$ has a basis $\{(1,1,\dots,1)\}\cup\{\Trp(\Ga_ix^j)\}_{1\le i\le 7,j=1,3,5}$, where $\{\Ga_1,\Ga_2,\dots,\Ga_7\}$ is an $\F_2$-basis of $\F_{128}$. Thus, $\Trp(V)$
is a binary $[128,22]$-linear code. To see the minimum distance, let $f$ be a nonzero polynomial of degree $m$ with $m\le 5$ and $\gcd(m,2)=1$. By the Weil-Serre bound (\ref{eq2.4}), we know that the number of zeros of $\Trp(f)$ can  be at most $64+22=86$, thus we get a lower bound on minimum distance, i.e., $d\ge 128-86=42$. However, by our bound in Theorem \ref{2.1}, we get a lower bound $d\ge 128-(64+16)=48$. This achieves the best-known bound \cite{G12}. In fact,  the software {\it Magma}  verifies that this code indeed has minimum distance $48$. In other words, our bound (\ref{eq2.5}) is tight in this case.
}\end{ex}
Next we study dual codes of  primitive BCH codes.
%Generalized Hamming weight was introduced for study of wire-tap channel of type II and the weight is used to characterize performance of a linear code in certain cryptographic application \cite{}. For an $\F_p$-vector space $W$ in $\F_p^N$, the the support of $W$ is defined by $\supp(W)=:\{i:\; \mbox{ there exists $(w_1,w_2,\dots,w_N)\in W$ such that $w_i\not=0$}\}$. Then the $r$th generalized Hamming weight of a $p$-ary linear code $C$ is defined to be
%\begin{equation}\label{eq2.11}d_r(C)=\min\{|\supp(W)|:\; \mbox{$W$ is a subspace of $C$ of dimension $r$}\}.\end{equation}

\begin{ex}\label{2.5}{\rm
Let $\Ga$ be a primitive element of $\F_q$ and let $\BCH(t)$ be a $t$-error correcting $p$-ary BCH code of length $N=q-1=p^n-1$. Then by Delsarte's Theorem \cite{MS77},  the dual $\BCH(t)^{\perp}$  can be represented as the trace code $\Trp(V)$ \cite{SV94}, where $\mP$ consists of all $q-1$ nonzero elements of $\F_q$, and $V$ is the $\F_q$-vector space generated by $\{1,x,x^2,\dots,x^{t}\}$. If $i$ is divisible by $p$, then $\Trp(x^i)=\Trp(x^{i/p})$. Thus,  $\BCH(t)^{\perp}$ is generated by $\{(1,1,\dots,1)\}\cup\{\Trp(\Ga_ix^j):\; 1\le j\le t, \ p\nmid j,\ 1\le i\le n\}$, where $\{\Ga_1,\Ga_2,\dots,\Ga_n\}$ is an $\F_p$-basis of $\F_q$. Hence, the dimension of $\BCH(t)^{\perp}$ is at most $1+n (t-\lfloor t/p\rfloor)$. On the other hand, if $t<\sqrt{q}/(p-1)$, then by the Weil-Serre bound (\ref{eq2.4}), $\Trp(f)\not=0$ for any $f\in V\setminus\{0\}$ with $\gcd(\deg(f),p)=1$. This implies that the dimension of $\BCH(t)^{\perp}$ is exactly $1+ n(t-\lfloor t/p\rfloor)$ for $t<\sqrt{q}/(p-1)$. By the bound (\ref{eq2.5}), we obtain a lower bound on minimum distance of $\BCH(t)^{\perp}$, namely
\begin{equation}\label{eq2.12}
d\left(\BCH(t)^{\perp}\right)\ge \left\{\begin{array}{ll}p^n-1-p^{n-1}-p^{\lceil n/g\rceil-1}\left\lfloor\frac{g\lfloor2\sqrt{q}\rfloor}{p^{\lceil n/g\rceil}}\right\rfloor&\mbox{if $g>1$}\\
p^n-1-p^{n-1}-p^{\lceil (n-1)/2\rceil}\left\lfloor\frac{\lfloor2\sqrt{q}\rfloor}{p^{ (n+1)/2}}\right\rfloor&\mbox{if $g=1$,}\end{array}
\right.
\end{equation}
where $g=(p-1)(t-1)/2$ if $p\nmid t$ and $g=(p-1)(t-2)/2$ if $p|t$¡£

Now we add all-one vector ${\bf 1}$ to $\BCH(t)^{\perp}$, then we get a code $C_p(t,n)$ generated by ${\bf 1}$ and  $\{\Trp(x^i):\; 0\le i\le 2t, \ p\not|i\}$. The dimension of $C$ increases by $1$, i.e., $1+n (t-\lfloor t/p\rfloor)$ for $t<\sqrt{q}/(p-1)$ and the lower bound (\ref{eq2.12}) on the minimum distance still holds for $C_p(t,n)$.

 Let us illustrate the parameters of our code $C$  by looking at some numerical results. Taking $p=2$, $t=4$ and $n=5$, we get a binary $[31,11,\ge 11]$-linear code $C_2(4,5)$. This is a best possible code in the sense that for given length $31$ and dimension $11$, the minimum distance can not be improved \cite{G12}.

 If $p=3$, $t=3$ and $n=3$, then $C_3(3,3)$ is a ternary $[26,7,\ge 14]$-linear code. It is also an optimal code \cite{G12}.
}\end{ex}

\begin{ex}\label{2.6}{\rm In this example, we consider the duals of classical Goppa codes.  Let $q=p^n$ and let $L=\F_q=\{\Ga_1,\Ga_2,\dots,\Ga_q\}$. Let $g(x)$ be a polynomial of degree $t$ with $\gcd(t,p)=1$ and $g(\Ga_i)\not=0$ for all $i=1,2,\dots,q$.
Then the classical Goppa code $\Gamma(L,g)$ defined by
 \[\Gamma(L,g)=\left\{(c_1,c_2,\dots,c_q)\in\F_p^n:\; \sum_{i=1}^q\frac{c_i}{x-\Ga_i}\equiv 0\bmod{g(x)}\right\}.\]
 is a $p$-ary linear code of length $q=p^n$. By Delsarte's Theorem \cite{MS77},  the dual $\Gamma(L,g)^{\perp}$  can be represented as the trace code $\{(\Tr(v_1f(\Ga_1)),(\Tr(v_1f(\Ga_1)),\dots,(\Tr(v_1f(\Ga_q))):\; f\in\F_q[x],\ \deg(f)\le t-1\}$, where $v_i=g(\Ga_i)$ \cite{SV94}. It is clear that $\Gamma(L,g)^{\perp}$ is equivalent to
$\Trp(V)$, where $\mP$ consists of all $q$ elements of $\F_q$, and $V$ is the $\F_q$-vector space generated by $\{1, x,x^2,\dots,x^{t-1}\}$. In the same way, we have that the dimension of $\Gamma(L,g)^{\perp}$ is exactly $1+n(t-\lfloor (t-1)/p\rfloor)$ for $t<\sqrt{q}/(p-1)$. By the bound (\ref{eq2.5}), we obtain a lower bound on minimum distance of $\Gamma(L,g)^{\perp}$, namely
\begin{equation}\label{eq2.13}
d\left(\Gamma(L,g)^{\perp}\right)\ge \left\{\begin{array}{ll}p^n-p^{n-1}-p^{\lceil n/g\rceil-1}\left\lfloor\frac{g\lfloor2\sqrt{q}\rfloor}{p^{\lceil n/g\rceil}}\right\rfloor&\mbox{if $g>1$}\\
p^n-p^{n-1}-p^{\lceil (n-1)/2\rceil}\left\lfloor\frac{\lfloor2\sqrt{q}\rfloor}{p^{ (n+1)/2}}\right\rfloor&\mbox{if $g=1$,}\end{array}
\right.
\end{equation}
where $g=(p-1)(t-2)/2$ if $p\nmid(t-1)$ and $g=(p-1)(t-3)/2$ if $p|(t-1)$

Taking $p=5$, $t=2$ and $n=3$, we get a $5$-ary $[125,7,\ge 95]$-linear code. This is a best possible code in the sense that for given length $125$ and dimension $17$, the minimum distance can not be improved \cite{G12}.
}\end{ex}

%\subsection{Locally testable codes}

\section{Outline of the proof of the main result}

The main objective of this section is to outline the proof of our Main Theorem \ref{2.1} in more elementary way.
\subsection{$L$-polynomial}
The Artin-Schreier type curve $\mX_f$ defined in (\ref{eq2.1}) is just an example of general algebraic curves. The curve $\mX_f$ is defined over $\F_q$ and of course it can  be regarded as a curve over arbitrary extension $\F_{q^i}$ for every $i\ge 1$. Let $N_f(i)$ denote the number of $\F_{q^i}$-rational points on $\mX_f$, i.e.,  $N_f(i)$ stands for the size of the set of $\F_{q^i}$-rational points
$\{\Pin\}\cup\{(\Ga,\Gb)\in\F_{q^i}:\; \Gb^p-\Gb=f(\Ga)\}.$
Then one can define the zeta function of $\mX_f$ by
$
\zeta_f(T):=\exp\left(\sum_{i=1}^{\infty}\frac{N_f(i)}{T^i}\right).
$
It was proved by Weil \cite{Weil48} that $\zeta_f(T)$ is a rational function of the form $\frac{L_f(T)}{(1-T)(1-qT)}$, where $L_f(T)\in\ZZ[T]$ is a polynomial of degree $2g=(m-1)(p-1)$. Furthermore, $L_f(0)=1$ and every reciprocity root of $L_f(T)$ has absolute value $\sqrt{q}$.

If we write $L_f(T)=\sum_{i=0}^{2g}a_iT^i$, then $a_0=L_f(0)=1$ and $a_{i+g}=q^ia_i$ for every all $0\le i\le g$. In particular, the leading coefficient $a_{2g}=q^g$.

\begin{definition}\label{3.1}{\rm
Let $L_f(T)=\sum_{i=0}^{2g}a_iT^i$ be the $L$-polynomial of $\mX_{f}$. Then the Hasse-Witt invariant of $\mX_f$, denoted by $i_{\mX_f}$, is defined to be the maximal $j$ such that $a_j\not\equiv 0\bmod{p}$, i.e.,
\begin{equation}\label{eq:001}i_{\mX_f}=\max\{0\le j\le g:\; a_j\not\equiv 0\bmod{p}\}.\end{equation}
}
\end{definition}
Since $a_0=1$, we have $i_{\mX_f}\ge 0$. On the other hand, we clearly have $i_{\mX_f}\le g$.

\begin{lem}\label{3.2} The Hasse-Witt invariant  $i_{\mX_f}$ for the curve $\mX_f$ defined in (\ref{eq:001}) is $0$.
\end{lem}
We refer to Lemma \ref{a.4} for the proof of Lemma \ref{3.2}.

\subsection{Newton polygon}
Denote by $\QQ_p$ the $p$-adic field. It is the completion field of the rational field $\QQ$ at prime $p$. The unique discrete valuation in $\QQ_p$ is again denoted by $\nu_p$.
The Newton polygon of a polynomial $u(x)\in\QQ_p[x]$ over the local field $\QQ_p$ provides information on factorization of $u(x)$ over $\QQ_p$. We briefly describe the  Newton polygon method in this subsection. The reader may refer to \cite[Section 3.1]{W63} for the details.

Let $u(x)=u_0+u_1x+\cdots+u_mx^m$ be a polynomial over $\QQ_p$ with $u_0u_m\neq 0$. For each $0\le i\le m$, we assign a point in $\RR^2$ as follows: (i) if $u_i\neq 0$, take the point $(i,\nu_p(u_i))$; (ii) if $u_i=0$, we ignore the point $(i,\infty)$. In this way, we form an envelope for the set of points $\{(i,\nu_p(u_i)):\; i=0,1,2,\dots,m\}$. The polygon thus determined is called the {\it Newton polygon}.

\begin{lem}\label{3.3} Suppose $(r,\nu_p(u_r))\leftrightarrow(s,\nu_p(u_s))$ with $s>r$ is a segment of the Newton polygon of $u(x)$ with slope $-k$. Then $u(x)$  has exactly $s-r$ roots $\Ga_1,\Ga_2,\dots,\Ga_{s-r}$ with $\nu_p(\Ga_1)=\nu_p(\Ga_2)=\cdots=\nu_p(\Ga_{s-r})=k$. Moreover, the polynomial $a(x):=\prod_{i=1}^{s-r}(x-\Ga_i)$ belongs to $\QQ_p[x]$.
\end{lem}
   Now if $h(x)\in\QQ_p[x]$ is an irreducible factor of $a(x)$ in the above lemma, then $\nu_p(h(0))=k\deg(h(x))\le k\deg(a(x))=k(s-r)$.

\subsection{Weil number}

\begin{definition}\label{3.4}{\rm
\begin{itemize}
\item[(i)]
 An algebraic number $\Go$ is called a $q$-Weil number if $\Go$ and all $\QQ$-conjugates of $\Go$ have absolute value $\sqrt{q}$ (by a $\QQ$-conjugate of $\Go$, we mean a root of the minimal polynomial of $\Go$ over $\QQ$).
\item[(ii)] A monic polynomial $\Phi(T)$ over $\ZZ[T]$ is called a $q$-Weil polynomial if it has an even degree and all its roots are $q$-Weil numbers.
\item[(iii)] The Hasse-Witt invariant of a $q$-Weil polynomial
$\Phi(T)=\sum_{i=0}^{2g}c_{2g-i}T^i\in\ZZ[t]$ is defined to be the maximal $j\in[0,g]$ such that $c_{j}\not\equiv 0\bmod{p}$.
\end{itemize}
}
\end{definition}

\begin{lem}\label{3.5} A $q$-Weil polynomial $\Phi(t)$ must have the form
\begin{equation}\label{eq:002}\sum_{i=0}^{2g}c_{2g-i}t^i\in\ZZ[t] \ \mbox{with $c_{2g}=q^g$ and $c_{2g-i}=q^{i}c_{i}$ for all $i=0,1,\dots,g$}.\end{equation}
\end{lem}
\begin{proof} First, note that the product of two  polynomials of the form (\ref{eq:002}) still has the form (\ref{eq:002}). Hence, it is sufficient to show that
 an irreducible $q$-Weil polynomial $\Phi(T)$ over $\ZZ$ has the form (\ref{eq:002}).  Let $\Go$ be a root of $\Phi(T)$. If $\Go$ is a real number, we must have $\Go=\sqrt{q}$ or $-\sqrt{q}$. Thus, $\Phi(T)=(T-\sqrt{q})(T+\sqrt{q})=T^2-q$. In this case, $\Phi(T)$ has the form (\ref{eq:002}).

  If $\Go$ is not a real number, we may assume that all $\QQ$-conjugates of $\Go$ are $\{\Go_1,\bar{\Go}_1,\dots, \Go_g,\bar{\Go}_g\}$, where $\bar{\Go}_i$ are complex conjugate of $\Go_i$. By definition of a $q$-Weil polynomial, we have $\Phi(T)=\prod_{i=1}^g(T-\Go_i)(T-\bar{\Go}_i)$ and $|\Go_i|=|\bar{\Go}_i|=\sqrt{q}$ for all $1\le i\le g$. The desired result follows from the following identity
  \[\frac{T^{2g}}{q^g}\Phi\left(\frac qT\right)=\frac{T^{2g}}{q^g}\prod_{i=1}^g\left(\frac qT-\Go_i\right)\left(\frac qT-\bar{\Go}_i\right)=
  \prod_{i=1}^g\left(T-\frac q{\Go_i}\right)\left(T-\frac q{\bar{\Go}_i}\right)=\Phi(T).\]
  Note that we use the fact that $\frac q{\Go_i}=\bar{\Go}_i$ in the above identity.
\end{proof}

\begin{lem}\label{3.6} Let $\Phi(T)$ is a $q$-Weil polynomial  with Hasse-Witt invariant equal to $0$ and let $\Psi(T)$ be a divisor of $\Phi(T)$ which is also a $q$-Weil polynomial. Then $\Psi(T)$ has Hasse-Witt invariant equal to $0$ as well.
\end{lem}
\begin{proof} Let $\Phi(T)=\sum_{i=0}^{2g}c_{2g-i}T^i\in\ZZ[t]$ and let $\Psi(T)=\sum_{i=0}^{2r}a_{2r-i}T^i\in\ZZ[T]$. Put $\Phi(T)/\Psi(T)=\sum_{i=0}^{2s}b_{2s-i}T^i\in\ZZ[T]$. Then $r+s=g$.

Suppose that $\Psi(T)$ has Hasse-Witt invariant bigger than $0$. Let $i$ be the largest index such that $\nu_p(a_i)=0$. Then $1\le i\le r$. Let $j$ be the largest index such that $\nu_p(b_j)=0$ for some $0\le j\le s$. Consider
\begin{equation}\label{eq3.2}c_{i+j}=\sum_{k+\ell=i+j}a_kb_{\ell}=a_ib_j+\sum_{k+\ell=i+j, (k,\ell)\not=(i,j)}a_kb_{\ell}.\end{equation}
Every term in the summation of the above equation (\ref{eq3.2}) is divisible by $p$, while $a_ib_j$ is not divisible by $p$. Thus, $c_{i+j}$ is not divisible by $p$. This is a contradiction to our condition.
\end{proof}

\subsection{Characteristic polynomial}

Let $L_f(T)=\sum_{i=0}^{2g}a_iT^i$ be the $L$-polynomial of $\mX_f$. By abuse of notation, the reciprocal polynomial $\tL_f(T)=\sum_{i=0}^{2g}T^{2g-i}$ of $L_f(T)$ is called the {\it characteristic polynomial} of $\mX_f$ (in fact, $\tL_f(T)$ is the characteristic polynomial of the Jacobian of the curve $\mX_f$ (see Appendix)). Then it is clear that every root of $\tL_f(T)$ is a Weil number and $\tL_f(T)$ is a $q$-Weil Polynomial.
\begin{lem}\label{3.7} If the characteristic polynomial $\tL_f(T)$ of $\mX_f$ is canonically factorized into product $\tL(T)=\prod r(T)^e$ of powers of irreducible polynomials over $\QQ$, then the power $e$ is the least common denominator of $\nu_p(h(0))/n$, where $h(T)$ runs through all irreducible factors of $r(T)$ in $\QQ_p[x]$.
\end{lem}
We refer Lemma \ref{3.7} to Theorem \ref{a.7}.

\begin{theorem}\label{3.8} Assume that the characteristic polynomial $\tL_f(T)$ of $\mX_f$ is canonically factorized into product $\tL(T)=\prod r(T)^e$ over $\QQ$.   Write $r(T)^e=\sum_{i=0}^{2g}b_{2g-i}T^i\in\ZZ[T]$. Then  $\nu_p(b_1)\ge n/g$  if $g$ is bigger than  $1$.
\end{theorem}
\begin{proof} We outline the proof here. We refer to Theorem \ref{a.8} for the detailed proof.

 Let $r(T)=\sum_{i=0}^{2r}a_{2g-i}T^i\in\ZZ[T]$.
 Since $e$ is the least common denominators of $\nu_p(h(0))/n$ for all irreducible factors $h(T)$ of  $r(T)$ over $\QQ_p$, we consider the Newton polygon of the polynomial $r(T)$. By identifying one particular segment $(2r-i,\nu_p(a_i))\leftrightarrow(2r,0)$ for some $1\le i\le r$ and its slope $-\nu_p(a_i)/i$, one obtains a lower bound on $\nu_p(a_i)$ for some $1\le i\le g$. Furthermore, the slope $-\nu_p(a_1)$ of the segment $(2r-1,\nu_p(a_1))\leftrightarrow(2r,0)$ is at most the slope $-\nu_p(a_i)/i$. This gives a lower bound on $\nu_p(a_1)$. The desired result follows from this lower bound on $\nu_p(a_1)$.
\end{proof}

\begin{theorem}\label{3.9} Let $\tL_f(T)=\sum_{i=0}^{2g}a_{2g-i}T^i\in\ZZ[T]$ be the characteristic polynomial of  the curve $\mX_f$ defined in (\ref{eq3.1}) with $g=(m-1)(p-1)/2$. Then we have
\[\nu_p(a_1)\ge\left\{\begin{array}{ll}
\left\lceil\frac ng\right\rceil & \mbox{if $g>1$}\\
\frac{n+1}2 & \mbox{if $g=1$}
\end{array}
\right.  \]
\end{theorem}
\begin{proof}
By  Lemma \ref{3.2}, we know that $\tL_f(T)$ is a $q$-Weil polynomial with Hasse-Witt invariant equal to $0$. If $g=1$, then we must have $p=2$ or $3$. In this case, the  desired result follows from \cite[Theorem 4.1]{W69}. Now assume that $g\ge 2$. Factorize $\tL_f(T)$ into a product $\prod_{i=1}^k\Phi_i(T)$ of co-prime $q$-Weil polynomials  such that every $\Phi_i(T)$ is a power of an irreducible polynomial over $\QQ$. Let $\Phi_i(T)=\sum_{j=0}^{2g_i}a_{i,2g_i-j}T^j$. Then we have $\nu_p(a_{i,1})\ge \min\{n/g_i,(n+1)/2\}\ge n/g$ for all $i=1,2,\dots,k$. Thus, we have
$\nu_p(a_1)=\nu_p\left(\sum_{i=1}^ka_{i,1}\right)\ge\min_{1\le i\le k}\{\nu_p(a_{i,1})\}\ge n/g.$
\end{proof}
{\bf Proof of the Main Result (Theorem \ref{2.1})}: Let $\tL_f(t)=\sum_{i=0}^{2g}a_{2g-i}T^i\in\ZZ[T]$ be the characteristic polynomial of  $\mX_f$  with $g=(m-1)(p-1)/2$. First we note that $|N_f-q-1|=|a_1|$. Combining Theorem \ref{3.9} and the Weil-Serre bound, we get
\[|N_g-q-1|\le\left\{\begin{array}{ll}
p^{\lceil n/g\rceil}\left\lfloor\frac{g\lfloor 2\sqrt{q}\rfloor}{p^{\lceil n/g\rceil}}\right\rfloor & \mbox{if $g>1$}\\
p^{(n+1)/2}\left\lfloor\frac{g\lfloor 2\sqrt{q}\rfloor}{p^{(n+1)/2}}\right\rfloor  & \mbox{if $g=1$.}
\end{array}
\right. \]
The desired result follows from the fact that $|Z_f|=(N_f-1)/p$.\epf

\newpage

\appendix
\begin{center}
{\bf Appendix}
\end{center}

\section{Algebraic curves and Hasse-Witt invariants}

Throughout the Appendix, we assume that $\mX$ is an absolutely irreducible, projective and smooth algebraic curve over $\F_q$ of genus $g$. Then it can be regarded as a curve over $\F_{q^i}$ for any $i\ge 1$. Denote by $N_{\mX}(i)$  the number of $\F_{q^i}$-rational points.
Then one can define the zeta function of $\mX$ by
\begin{equation}\label{eq:004}
\zeta_{\mX}(T):=\exp\left(\sum_{i=1}^{\infty}\frac{N_{\mX}(i)}{T^i}\right).
\end{equation}
It was proved by Weil \cite{Weil48} that $\zeta_{\mX}(T)$ is a rational function of the form $\frac{L_{\mX}(T)}{(1-T)(1-qT)}$, where $L_{\mX}(T)\in\ZZ[T]$ is a polynomial of degree $2g$. Furthermore, $L_{\mX}(0)=1$ and every reciprocity root of $L_f(T)$ has absolute value $\sqrt{q}$.

If we write $L_{\mX}(T)=\sum_{i=0}^{2g}a_iT^i$, then $a_0=L_f(0)=1$ and $a_{i+g}=q^ia_i$ for every all $0\le i\le g$. In particular, the leading coefficient $a_{2g}=q^g$.

We need some preliminaries and background on algebraic geometry, in particular, abelian varieties in this Appendix. The reader may refer to \cite{Sha95,W69,Mum08} for some basic results on curves and abelian varieties.

In Section 3.1, we defined the Hasse-Witt invariant of $\X$ through the $L$-polynomial of $\mX$. Actually, it is more proper to define it in terms of $p$-torsion points of its Jacobian.
For an algebraic curve $\X$ over $\F_q$ of genus $g$, we denote by $\mJ_{\mX}$ the Jacobian of $\mX$. Then $\mJ_{\mX}$ is an abelian variety over $\F_q$ of dimension $g$.  Let $\mJ_{\mX}(\F_q)$ and $\mJ_{\mX}(\oF_q)$ be the groups of the $\F_q$-rational points and $\oF_q$-rational points on $\mX$, respectively.
\begin{definition}\label{a.1}
Denote by $\mJ_{\mX}[p]$ the subgroup of the $p$-torsion points of $\mJ_{\mX}(\oF_q)$, then one has $|\mJ_{\mX}[p]|= p^r$ for some integer $r$ with $0\le r\le g$. The integer $r$ is called the Hasse-Witt invariant of $\X$, denoted by $i_{\mX}$.
\end{definition}

 For each abelian variety $A$ over $\F_q$, we can consider the Frobenius morphism $\pi_A$ which sends every coordinate of a point on this variety to its $q$th power. The characteristic polynomial $\Phi_A$ of $\pi_A$  over the local field $\QQ_p$ is called the characteristic polynomial of $A$ (see \cite{W69,Mum08}). It is, in fact, a polynomial over $\ZZ$. Furthermore, if $A$ is the Jacobian of an algebraic curve $\X/\F_q$, then the characteristic polynomial of $A$ is in fact the reciprocal polynomial of the $L$-polynomial of $\mX$.

The following result given in \cite{S79} characterizes $i_{\mX}$ in terms of the characteristic polynomial of $\mJ_{\mX}$.

\begin{prop}\label{a.2} Let $\Phi_{\mX}(t)=\sum_{i=0}^{2g}a_{2g-i}t^i$ be the characteristic polynomial of $\mJ_{\mX}$. Then $i_{\mX}$ is the maximal $j$ such that $a_j\not\equiv 0\bmod{p}$, i.e.,
\[i_{\mX}=\max\{0\le j\le g:\; a_j\not\equiv 0\bmod{p}\}.\]
\end{prop}

\begin{theorem}[Deuring-Shafarevich (see e.g.~\cite{HKT08})]\label{a.3} Let $E,F$ be the function fields of two curves $\mX/\oF_q$ and $\mY/\oF_q$, respectively.
Assume that $E/F$ is a Galois extension of function fields and the Galois group of this
extension is a $p$-group. Then
$$i_{\mX}-1=[E:F](i_{\mY}-1)+\sum_{P\in\mY}\sum_{\substack{Q\in\mX \\\  Q|P}}(e(Q|P)-1).$$
\end{theorem}

\begin{lem}\label{a.4} The Hasse-Witt invariant  $i_{\mX_f}$ for the curve $\mX_f$ defined in (\ref{eq2.1}) is $0$.
\end{lem}
\begin{proof} Let $F$ be the rational function field $\oF_q(x)$ of the projective line and let $E$ be the function field $\oF_q(\mX_f)$ of $\mX_f$. The only ramified point is the unique point $\infty$  lying over the pole of $x$. Moreover, the ramification index is $p$. Thus, by the Deuring-Shafarevich Theorem, we have
\[i_{\mX_f}-1=[E:F](0-1)+p-1=p\times (0-1)+p-1=-1,\]
i.e., $i_{\mX_f}=0$.
\end{proof}
\section{Abelian varieties}
We introduce a few definitions now.

\begin{definition}\label{a.5}{\rm
\begin{itemize}
  \item[(i)] An abelian variety is called elementary or simple if it has no subvarieties.
\item[(ii)] Two abelian varieties over $\F_q$ are said { isogenous} if they have the same characteristic polynomial.
\end{itemize}
}
\end{definition}
The  characteristic polynomial of an abelian variety over $\F_q$ is a $q$-Weil polynomial. Thus, by Lemma \ref{3.5} the  characteristic polynomial of an abelian variety over $\F_q$ of dimension $g$ has the form $\sum_{i=0}^{2g}a_{2g-i}T^i\in\ZZ[T]$ with $a_0=1$ and $a_{2g-i}=q^{g-i}a_i$ for $i=0,1,\dots,g$.
The following well-known result (see \cite{W69}) provides information on factorization of characteristic polynomials of an abelian varieties.
\begin{theorem}\label{a.6} Every abelian variety is isogenous to a product of elementary abelian varieties, i.e., the characteristic polynomial of every abelian variety is a product of characteristic polynomials of elementary abelian varieties.
\end{theorem}
In view of the above theorem, it is sufficient to consider elementary abelian varieties in our case (see \cite{W69}).

\begin{theorem}[Tate-Honda]\label{a.7} There is one-to-one correspondence between isogeny classes of elementary abelian varieties over $\F_q$ and conjugacy classes of $q$-Weil numbers. More precisely, a polynomial $\Phi(T)$ is the characteristic polynomial of an elementary abelian variety over $\F_q$ if and only if $\Phi(T)=r(T)^e$ for some irreducible polynomial $r(T)\in\ZZ[T]$ with all roots being $q$-Weil numbers and $e$ is the least common denominators of $\nu_p(h(0))/n$, where $h(T)$ runs through all irreducible factors over $\QQ_p$.
\end{theorem}

\section{Abelian varieties with Hasse-Witt invariants $0$}

The following result plays a crucial role.

\begin{theorem}\label{a.8} Let $\Phi_A(T)=\sum_{i=0}^{2g}b_{2g-i}T^i\in\ZZ[T]$ be the  characteristic polynomial of an elementary abelian variety $A$ with Hasse-Witt invariant equal to $0$. If the dimension $g$ of $A$ is bigger than  $1$, then  $\nu_p(b_1)\ge n/g$.
\end{theorem}
\begin{proof} By Theorem \ref{a.7}, we know that every elementary abelian variety $A$ has the characteristic polynomial of the form  $\Phi_A(T)=r(T)^e$ for an integer $e\ge 1$ and a monic irreducible polynomial $r(T)$ over $\ZZ$.

If $r(T)$ has a real root, then  this root must be $\pm\sqrt{q}$. In this case, we have that $r(T)=(T-\sqrt{q})(T+\sqrt{q})=T^2-q$. Hence, $b_1=0$ and the desired result follows.

 Now we assume that all roots of $r(T)$ are not real. Then it is clear that the degree of $r(T)$ is even. Let $r(T)=\sum_{i=0}^{2r}a_{2r-i}T^i$ with $g=er$. To analyze $\nu_p(a_i)$, we look at the Newton polygon  of $r(T)$.
 By our assumption, we know that $\nu_p(a_i)>0$ for all $1\le i\le r$.

Define the set
\[I:=\{1\le i\le r:\;  \nu_p(a_i)\ge in/2\}\] and
\[J:=\{1\le j\le r:\;  \nu_p(a_j)< jn/2\}.\]
If $1\in I$, then  $\nu_p(a_1)\ge n/2\ge n/g$. Hence,  $\nu_p(b_1)=\nu_P(e)+\nu_p(a_1)\ge n/g$.

Now assume that $1\in J$. We claim the following.
\begin{quote}
There exists $1\le i\le r$ such that
\begin{equation}\label{eq3.3}\frac{\nu_p(a_{i})}{i}<\min_{r+1\le \ell\le 2r}\left\{\frac{\nu_p(a_{\ell})}{\ell}\right\}.\end{equation}
\end{quote}

First of all, $\nu_p(a_1)<n/2=\nu_p(q^r)/2r=\nu_p(a_{2r})/2r$. Hence, to prove the inequality (\ref{eq3.3}), it is sufficient to prove that for every $\ell\in\{r+1,r+2,\dots,2r-1\}$, there exists $j\in J$ such that $\nu_p(a_{\ell})/\ell>{\nu_p(a_{j})}/{j}$. Note that, since $I\cup J=\{1,2,\dots,r\}$, thus $\ell=2r-k$ with $k$ in $I\setminus\{r\}$ or in $J$.

 Then for every $i\in I\setminus\{r\}$ and $j\in J$, we have
\[\frac{\nu_p(a_{2r-j})}{2r-j}= \frac{\nu_p(a_{j})+(r-j)n}{2r-j}>\frac{\nu_p(a_{j})}{j}\]
and
\[\frac{\nu_p(a_{2r-i})}{2r-i}= \frac{\nu_p(a_{i})+(r-i)n}{2r-i}\ge \frac n2>\frac{\nu_p(a_{j})}{j}.\]
Our claim follows.

Now let $0\le i\le r$ be the largest index such that
\[\frac{\nu_p(a_{i})}{i}=\min_{1\le j\le 2r}\left\{\frac{\nu_p(a_{j})}{j}\right\},\]
 then $(2r-i,\nu_p(a_i))\longleftrightarrow(2r,0)$ forms a segment in the Newton Polygon of $r(T)$ with slope $-\nu_p(a_i)/i$. Assume that $h(t)$ is an irreducible factor of $r(T)$ in $\QQ_p[t]$ corresponding to this segment. Then we have $\nu_p(h(0))\le \frac{\nu_p(a_i)}{i}\times (2r-(2r-i))=\nu_p(a_i)$.  Assume that $\frac{\nu_p(h(0))}n=\frac{k}{\ell}$ with $\gcd(k,\ell)=1$. Then we have $e\ge \ell$ and
\begin{equation}\nu_p(a_i)\ge\nu_p(h(0))=n\times \frac{\nu_p(h(0))}n=\frac{kn}{\ell}\ge \frac{kn}{e}\ge \frac{n}{e}.\end{equation}
Since $(2r-i,\nu_p(a_i))\longleftrightarrow(2r,0)$ forms a segment in the Newton Polygon of $h(t)$, the slope $-\nu_p(a_1)$ of the segment $(2r-1,\nu_p(a_1))\longleftrightarrow(2r,0)$ is at most the slope $-\nu_p(a_i)/i$ of the segment $(2r-i,\nu_p(a_i))\longleftrightarrow(2r,0)$, i.e.,
\[\nu_p(a_1)\ge \frac{\nu_p(a_i)}i\ge \frac{n}{ie}\ge \frac{n}{re}=\frac{n}{g}.\]
Hence, $\nu_p(b_1)\ge \nu_p(e)+\nu_p(a_1)\ge \frac{n}{g}$.
This completes the proof.
\end{proof}

\begin{cor}\label{a.8} Assume that $\tL(T)$ is the reciprocal polynomial of the $L$-polynomial of an algebraic curve $\mX$ over $\F_q$. Let $\tL(T)$ have the canonical factorization into product $\prod r(T)^e$. If the Hasse-Witt invariant of $\mX$ is $0$, then $\nu_P(a_1)\ge n/g$, where $ r(T)^e=\sum_{i=0}^{2g}a_{2g-i}T^i$ with $g>1$.

\end{cor}

\end{document}